\documentclass[pra,onecolumn,
showpacs,
preprintnumbers,amsmath,amssymb,superscriptaddress]{revtex4}

\usepackage{bm}
\usepackage{graphicx}
\usepackage{amsbsy} 
\usepackage{amsmath}
\usepackage{amsfonts}
\usepackage{amsthm}
\usepackage{bm}
\usepackage{graphicx}
\usepackage{amsbsy}
\usepackage{amsmath}
\usepackage{amsfonts}
\usepackage{amsthm}
\usepackage{braket}
\usepackage{dsfont}
\usepackage{color}
\usepackage{mathrsfs}

\newcommand{\inn}[2]{{\left \langle #1 | #2 \right\rangle}}

\newcommand{\T}{\operatorname{Tr}}

\renewcommand{\labelenumi}{(\roman{enumi})}

\theoremstyle{plain}
\newtheorem{theorem}{Theorem}
\newtheorem{lemma}[theorem]{Lemma}

\theoremstyle{definition}
\newtheorem{definition}{Definition}

\theoremstyle{remark}
\newtheorem*{remark}{Remark}

\begin{document}

\title{Stability theorem of depolarizing channels 
\\for the minimal output quantum R\'{e}nyi entropies}

\author{Eunok Bae}
\affiliation{Department of Mathematics and Research Institute for Basic Sciences,
 Kyung Hee University, Seoul 130-701, Korea}
 \affiliation{Institute for Quantum Science and Technology, 
 University of Calgary, Alberta, T2N 1N4, Canada}
\author{Gilad Gour}
\affiliation{Institute for Quantum Science and Technology, 
 University of Calgary, Alberta, T2N 1N4, Canada}
\affiliation{Department of Mathematics and Statistics, 
 University of Calgary, Calgary, Alberta T2N 1N4, Canada} 
\author{Soojoon Lee}
\affiliation{Department of Mathematics and Research Institute for Basic Sciences,
 Kyung Hee University, Seoul 130-701, Korea}
\author{Jeonghoon Park}
\affiliation{Department of Mathematics and Research Institute for Basic Sciences,
 Kyung Hee University, Seoul 130-701, Korea}
\author{Barry C. Sanders}
\affiliation{Institute for Quantum Science and Technology, 
 University of Calgary, Alberta, T2N 1N4, Canada}
\affiliation{Department of Physics and Astronomy, 
 University of Calgary, Calgary, Alberta T2N 1N4, Canada} \affiliation{Program in Quantum Information Science, Canadian Institute for Advanced Research,
 Toronto, Ontario, M5G 1Z8, Canada}
 \affiliation{%
Hefei National Laboratory for Physical Sciences at Microscale and Department of Modern Physics, University of Science and Technology of China, Hefei, Anhui 230026, China}
\affiliation{%
Shanghai Branch, CAS Center for Excellence and Synergetic Innovation Center in Quantum Information and Quantum Physics, University of Science and Technology of China, Shanghai 201315, China%
}

\date{\today}

\begin{abstract}
We show that the stability theorem of the depolarizing channel holds 
for the output quantum $p$-R\'{e}nyi entropy for $p \ge 2$ 
or $p=1$, 
which is an extension of the well known case $p=2$.
As an application, we present a protocol in which Bob determines 
whether Alice prepares a pure quantum state close to a product state.
In the protocol, Alice transmits to Bob multiple copies of a pure state
through a depolarizing channel, and Bob
estimates its output quantum $p$-R\'{e}nyi entropy.
By using our stability theorem, we show that Bob can determine whether her preparation is appropriate.
\end{abstract}


\maketitle

\section{Introduction}
\label{sec:Intro}
We extend the stability theorem of the depolarizing channel 
to the output quantum $p$-R\'{e}nyi entropy 
for $p \ge 2$ or $p=1$. 
The original stability theorem with the output purity is essentially equivalent to our stability theorem for the case $p=2$ 
and was used in proving the equality 
$\mathrm{QMA}(k)=\mathrm{QMA}(2)$ for all $k \ge 2$~\cite{HM13}. 
We generalize it to the output quantum $p$-R\'{e}nyi entropy 
to create a more powerful tool,
and we apply it to a type of polygraph test as discussed below.
Generalization is accomplished by defining the notion of stability of a quantum channel with respect to any real valued continuous function.
That is, if a state is close to achieving 
the minimal/maximal output value of a particular quantity (entropy function) through the channel,
then it must be close to an input state giving the minimal/maximal value.
In particular, we show that the depolarizing channel is stable 
with respect to the output quantum R\'{e}nyi entropy.

Our theorem is constructed by generalizing the Taylor expansion of von~Neumann entropy~\cite{GF13} to the quantum R\'{e}nyi entropy.
Whereas the original work employed the output purity~\cite{HM13},
which is a relatively simpler function,
we use a more general (and complicated) function,
namely the output quantum R\'{e}nyi entropy.
The Taylor expansion of the output quantum R\'{e}nyi entropy is the technique we use
to prove the stability theorem for a depolarizing channel with respect to the the output quantum R\'{e}nyi entropy.
The protocol is described in \S\ref{sec:Poly} 
and provides us meaning and intuition for 
the stability theorem of the depolarizing channel.
Furthermore, the protocol shows that our stability theorem has a benefit 
as our protocol has a smaller undecidable gap than the original case. 

We organize our paper as follows. 
In \S\ref{sec:Stable}, we provide some notions to define a stable channel clearly.
Our main result appears in \S\ref{sec:Stability}
where we generalize the Taylor expansion of the von~Neumann entropy 
to calculate the Taylor expansions of the quantum R\'{e}nyi entropies.
We use this result to show that the depolarizing channel is stable with respect to the output quantum $p$-R\'{e}nyi entropies for $p \ge 2$ or $p=1$. 
In \S\ref{sec:Poly}, 
we introduce a polygraph test as an application of our stability theorem. 
Finally, in \S\ref{sec:Conclusions}, we conclude with discussion on our results.

\section{Stable channels}
\label{sec:Stable}
In this section, we define the notions of 
a {\em quantity}, an {\em extremal} state, an {\em~$\epsilon$-almost extremal} state, 
an {\em~$\epsilon$-stable} channel, and a {\em stable} channel. 
Except for this section, in the rest of the paper, we will be using these notions for only 
the depolarizing channel and the quantum R\'{e}nyi entropy as a quantity.

\begin{definition}
Let 
$\mathcal{E}: \mathcal{B}(\mathcal{H}_{\text i}) \rightarrow \mathcal{B}(\mathcal{H}_{\text o})$ 
be a quantum channel (i.e., a trace preserving completely positive map) and let $Q$ be a real-valued continuous function on $\mathcal{B}(\mathcal{H}_{\text o})$, 
where $\mathcal{B}(\mathcal{H}_{\text i})$ and $\mathcal{B}(\mathcal{H}_{\text o})$ are the sets of all states in the input space $\mathcal{H}_{\text i}$ and the output space $\mathcal{H}_{\text o}$, respectively.

For any $\epsilon>0$, 
a state $\sigma \in \mathcal{B}(\mathcal{H}_{\text i})$ is {\em~$\epsilon$-almost extremal}
with respect to the function $Q$ and the channel $\mathcal{E}$ 
if
\begin{equation}
	\left| Q\left(\mathcal{E}(\sigma)\right)-\operatorname{ext}_\rho Q\left( \mathcal{E} \left(\rho \right) \right) \right|
		\in O(\epsilon),
\end{equation} 
where the extremal value, $\operatorname{ext}_\rho$, refers to 
either the maximal value or the minimal value of $Q$ over all states $\rho$ in $\mathcal{B}(\mathcal{H}_{\text i})$ 
according to a given quantity.
A state $\sigma_0$ is said to be {\em extremal} 
with respect to $Q$ and $\mathcal{E}$ if 
\begin{equation}
	Q\left(\mathcal{E}(\sigma_0)\right)=\operatorname{ext}_\rho Q \left( \mathcal{E} \left(\rho\right) \right).
\end{equation}
\end{definition}
Now we define a stable channel with respect to the function $Q$.
\begin{definition}
For a given $\epsilon>0$, a channel $\mathcal{E}$ is {\em~$\epsilon$-stable} 
with respect to a quantity $Q$
if, for all $\sigma$~$\epsilon$-almost extremal, 
an extremal state $\sigma_0$ exists with respect to $Q$ and $\mathcal{E}$ such that
\begin{equation}
	\left\| \sigma- \sigma_0 \right\|^2_1 \in O(\epsilon),
\end{equation}
where $\left\| \cdot \right\|_1$ denotes the trace norm.
A channel $\mathcal{E}$ is {\em stable} with respect to a quantity $Q$
if it is~$\epsilon$-stable with respect to the quantity $Q$ 
for all $\epsilon>0$.
\end{definition}

We have provided some generalized definitions to establish the notion of a stable channel.
In the next section, as our main result, we present the stability theorem of the depolarizing channel for the output quantum R\'{e}nyi entropy and prove that the depolarizing channel is stable with respect to the quantum R\'{e}nyi entropy.

\section{Stability of the depolarizing channel}
\label{sec:Stability}
In this subsection, we present and prove the stability theorem of the depolarizing channel with respect to the output quantum R\'{e}nyi entropy. This section consists of two subsections. In the first subsection, we evaluate the Taylor expansion of the quantum R\'{e}nyi entropy which is crucial to prove our main theorem in the second subsection.

\subsection{The Taylor expansion of the quantum $p$-R\'{e}nyi entropy}
\label{subsec:Taylor}
In this subsection, 
we the Taylor expansion technique for the von~Neumann entropy~\cite{GF13}
to calculate the Taylor expansion of the quantum $p$-R\'{e}nyi entropy.
This technique is key to prove the stability theorem of the depolarizing channel for the output quantum $p$-R\'{e}nyi entropy.

For $p > 0$ ($p\neq1$),
the quantum $p$-R\'{e}nyi entropy~\cite{H3} of a state~$\rho$ is
\begin{equation}
	S_p\left(\rho\right):=\frac{1}{1-p} \log\T\rho^{p}.
\label{eq:Renyi}
\end{equation}
The minimal output quantum $p$-R\'{e}nyi entropy of a quantum channel $\mathcal{E}$ 
is defined as
\begin{equation}
	S_p^{\min}\left(\mathcal{E}\right):=\min_\rho S_p\left(\mathcal{E}(\rho)\right),
\label{eq:minRenyi}
\end{equation}
where the minimum is taken over all input states $\rho$ of $\mathcal{E}$.
The quantum $p$-R\'{e}nyi entropy converges to the von~Neumann entropy as $p$ tends to one, 
and we can thus consider the quantum R\'{e}nyi entropy 
as a generalization of the von~Neumann entropy~\cite{H3}. 

In order to obtain the Taylor expansion of the quantum R\'{e}nyi entropy, 
we exploit the following lemma.
\begin{lemma}[Gour and Friedland~\cite{GF13}]
\label{lem:DividedDifference}
	Let $A=diag\left(p_1,\cdots,p_m\right) \in \mathbb{C}^{m \times m}$ be a diagonal square matrix, and $B=\left[b_{ij} \right] \in \mathbb{C}^{m \times m}$ be a complex square matrix. 
Let $f$ be a $C^2$ function defined on a real open interval $\left(a,b\right)$.
Then
\begin{equation}
	f\left(A+tB\right)
		=f\left(A\right)+t L_A\left(B\right)+t^2 Q_A\left(B\right)+O\left(t^3\right)
\label{eq:Lemma_GF}
\end{equation}
for $L_A : \mathbb{C}^{m \times m} \rightarrow \mathbb{C}^{m \times m}$ a linear operator
and $Q_A:\mathbb{C}^{m \times m} \rightarrow \mathbb{C}^{m \times m} $ 
a quadratic homogeneous non-commutative polynomial in $B$. 
For $i,j=1,\cdots,m$, we have
\begin{align}
&\left[L_A\left(B\right)\right]_{ij}
=\Delta f\left(p_i, p_j\right)b_{ij}
=\frac{f\left(p_i\right)-f\left(p_j\right)}{p_i-p_j}b_{ij}, 
\nonumber \\
&\left[Q_A\left(B\right)\right]_{ij}
=\sum_{k=1}^m\Delta^2 f\left(p_i, p_k, p_j\right)b_{ik}b_{kj}.
\label{eq:lemma_GF1}
\end{align}
In particular,
\begin{align}
&\T\left(L_A\left(B\right)\right)
=\sum_{j=1}^m f^\prime\left(p_j\right)b_{jj}, 
\nonumber \\
&\T\left(Q_A\left(B\right)\right)
=\sum_{i,j=1}^m \frac{f^\prime\left(p_i\right)-f^\prime\left(p_j\right)}
{2\left(p_i-p_j\right)}b_{ij}b_{ji}.
\label{eq:lemma_GF2}
\end{align}
\end{lemma}
Now we use Lemma~\ref{lem:DividedDifference} 
to calculate the Taylor expansion of $S_p\left(\rho\left(t\right)\right)$.
\begin{theorem}
\label{thm:TER}
A nonsingular density matrix
\begin{equation}
	\rho\left(t\right)=\rho+ t \gamma_0+t^2 \gamma_1+ O(t^3),
\end{equation}
with~$\rho$ diagonal, $\gamma_0$ all zeroes along the diagonal 
and~$\gamma_1$ having zero trace,
has quantum $p$-R\'{e}nyi entropy
\begin{align}
	S_p\left(\rho\left(t\right)\right)
=&S_p\left(\rho\right) +\frac{1}{1-p}t^2\left(p\frac{\T\left(\rho^{p-1}\gamma_1\right)}{\T\left(\rho^{p}\right)}+\frac{\T\left(Q_{\rho}\left(\gamma_0\right)\right)}{\T\left(\rho^{p}\right)}
\right)+O\left(t^3\right).
\label{eq:Salphat}
\end{align}
\end{theorem}
\begin{remark}
As $p$ tends to one, 
Theorem~\ref{thm:TER} 
implies the Taylor expansion of the von~Neumann entropy, 
hence generalizes the von~Neumann entropy.
\end{remark}
\begin{proof}
As $\rho$ is nonsingular, $\mathds{1}-\rho\left(t\right) < \mathds{1}$ for small~$t$. 
Thus, we can employ the Taylor expansion with respect to~$t$. 
From the following Taylor expansion
\begin{equation}
	\rho^{p}\left(t\right)
		=\left[\mathds{1}-\left(\mathds{1}-\rho\left(t\right)\right)\right]^{p}
=\sum_{n=0}^{\infty} \binom{p}{n}\left(-1\right)^n\left(\mathds{1}-\rho\left(t\right)\right)^n,
\label{eq:TER1}
\end{equation}
we obtain
\begin{equation}
	\T\left(\rho^{p}\left(t\right)\right)
		=\sum_{n=0}^{\infty} \binom{p}{n}\left(-1\right)^n
\T\left[\left(\mathds{1}-\rho\left(t\right)\right)^n\right].
\label{eq:TER2}
\end{equation}
Expanding the trace term in the right-hand side of Eq.~(\ref{eq:TER2}) 
up to second order in~$t$ yields
\begin{equation}
	\T\left[\left(\mathds{1}-\rho\left(t\right)\right)^n\right]
		=\T\left[\left(\mathds{1}-\sigma\left(t\right)\right)^n\right]-t^2 n 
\T\left[\left(\mathds{1}-\rho\right)^{n-1}\gamma_1\right]+O\left(t^3\right),
\label{eq:TER3}
\end{equation}
where $\sigma\left(t\right)=\rho+t\gamma_0$.
From Eq.~(\ref{eq:TER2}) and Eq.~(\ref{eq:TER3}),
\begin{equation}
\T\left(\rho^{p}\left(t\right)\right)
=\T\left(\sigma^{p}\left(t\right)\right)
+p t^2 \T\left( \rho^{p-1}\gamma_1\right)
+O\left( t^3\right).
\label{eq:TER4}
\end{equation}
As Lemma~\ref{lem:DividedDifference} yields the equality
\begin{equation}
	\T\left(\sigma^{p}\left(t\right)\right)
		=\T\left( \rho^{p}\right)+t p \T\left(\rho^{p-1}\gamma_0\right)
			+t^2 \T\left( Q_{\rho}\left( \gamma_0\right)\right)+O\left( t^3\right),
\label{eq:TER5}
\end{equation}
and $\gamma_0$ is zero along the diagonal, 
we obtain
\begin{equation}
	\T\left(\rho^{p}\left(t\right)\right)
		=\T\left(\rho^{p}\right)+t^2\left( p \T\left(\rho^{p-1}\gamma_1\right)
			+\T\left( Q_{\rho}\left( \gamma_0\right)\right)\right)+O\left( t^3\right).
\label{eq:TER6}
\end{equation}
Using the Taylor expansion of the logarithm function,
\begin{equation}
\log\left(1+x\right)=\sum_{n=1}^{\infty} (-1)^{n+1}\frac{x^n}{n},
\label{eq:TER7}
\end{equation}
we obtain
\begin{align}
	\log\T\left(\rho^{p}\left(t\right)\right)
		=&\log\left[\T\left(\rho^{p}\right)\left( 1+ t^2\left( p \frac{\T
			\left(\rho^{p-1}\gamma_1\right)}{\T\left(\rho^{p}\right)}
				+\frac{\T\left(Q_{\rho}\left(\gamma_0\right)\right)}
					{\T\left(\rho^{p}\right)}\right) +O\left( t^3\right) \right) \right]
						\nonumber\\
		=&\log\T\left(\rho^{p}\right)+ t^2\left( p \frac{\T\left(\rho^{p-1}
			\gamma_1\right)}{\T\left(\rho^{p}\right)}
				+\frac{\T\left(Q_{\rho}\left(\gamma_0\right)\right)}{\T\left(\rho^{p}\right)}\right)
					+O\left( t^3\right).
\label{eq:TER8}			
\end{align}
Therefore, by definition of the quantum $p$-R\'{e}nyi entropy in Eq.~(\ref{eq:Renyi}), 
the equality~(\ref{eq:Salphat}) can be readily obtained from Eq.~(\ref{eq:TER8}).
This completes the proof.
\end{proof}
We have evaluated the Taylor expansion of the quantum R\'{e}nyi entropy.
In the next subsection, we use this result to prove the stability theorem of the depolarizing channel for the output quantum R\'{e}nyi entropy.

\subsection{The stability theorem of the depolarizing channel for the output quantum R\'{e}nyi entropy}
\label{subsec:Stability}

In this subsection, we prove our main theorem, 
namely the stability theorem of the output quantum $p$-R\'{e}nyi entropy for the depolarizing channel 
for $p \ge 2$. 
First, we present the following lemma, 
which is crucial to prove the theorem. 
\begin{lemma}
\label{lem:f_p(x)}
For $p \ge 2$, $r> 1$ and $d \ge 2$, 
\begin{equation}
	 f_p(x)
	 	:= \frac{p}{1-p}\left[\left(\frac{({r}^{x})^{p-1}-1}{r^x-1}\right)
			\left(\frac{(r-1)^2}{(d+r-1)^2(r^p+(d-1))}\right)^{x}
+\left(\frac{r^{p-1}+r+(d-2)}{r^p+(d-1)}\right)^{x} -1 \right]
\label{eq:falphax}
\end{equation}
is monotonically increasing on $[2,\infty)$.
\end{lemma}
\begin{remark}
Let  $\ket{\psi}$ be an $n$-qudit pure state satisfying $\left| \inn{\psi}{\phi}\right|^{2}=1-t^2$
for an $n$-qudit product state $\ket{\phi}$. Then the function $f_p$~(\ref{eq:falphax}) is the coefficient of the second order term 
in the Taylor expansion of $S_p\left(\mathcal{D}_{\lambda}^{\otimes n}\ket{\psi}\bra{\psi}\right)$.

\end{remark}
\begin{proof}
\label{proof:f_p(x)}
Observe that
\begin{equation}
f'_p(x)=A_p\left[ B_p \log\left(\frac{r^p+d-1}{r^{p-1}(r-1)^2}\right)+C_p\log r \right],
\label{eq:f_p0}
\end{equation}
where
\begin{align}
A_p=&\frac{p}{1-p}\frac{(r-1)^2}{(r^x-1)^2(r^p+d-1)}, \nonumber \\
B_p=&(r^x-1)(1-(r^x)^{p-1}), \nonumber \\
C_p=&-(r^x)^p+p(r^x-1)+1. 
\label{eq:f_p1}
\end{align}
As $r>1$, straightforward calculations yield $A_p \leq 0$, $C_p \leq B_p \leq 0$ for $p \geq 2$.
Thus, we obtain the inequality 
\begin{align}
	f'_p(x) 
		\ge& A_p \left[ B_p\log\left(\frac{r^p+d-1}{r^{p-1}(r-1)^2}\right)+B_p \log r \right] 
		\nonumber\\
		=& A_p B_p \log\left(\frac{r^{p+1}+(d-1)r}{r^{p-1}(r-1)^2}\right).	
\label{eq:f_p2}
\end{align}
Here the right-hand side of the inequality (\ref{eq:f_p2}) is clearly nonnegative 
as the inequality
\begin{equation}
	\log\left(\frac{r^{p+1}+(d-1)r}{r^{p-1}(r-1)^2}\right) >0
	\label{eq:f_p3}
\end{equation}
		can be easily proved due to the inequality
\begin{equation}
	2r^{p-1}-r^{p-2}+(d-1) >0. 
	\label{eq:f_p4}
\end{equation}
Therefore, the function $f_p(x)$ is monotonically increasing.
\end{proof}

We now present one more lemma,
which tells us that the minimal output quantum $p$-R\'{e}nyi entropy of the depolarizing channel is achieved for product state inputs. 
%
The lemma can be readily obtained 
by the additivity of the minimal output quantum $p$-R\'{e}nyi entropy~\cite{Kin03}.

\begin{lemma}
\label{lem:minRenyi}
For the $n$-partite product depolarizing channel~$\mathcal{D}_{\lambda}^{\otimes n}$,
the quantum $p$-R\'{e}nyi entropy of the output state is minimized for product state inputs
and furthermore has the same value for all product state inputs;
that is, 
\begin{equation}
	S_p^{\min}\left( \mathcal{D}_{\lambda}^{\otimes n}\right)
		=S_p\left(\mathcal{D}_{\lambda}^{\otimes n}\ket{\phi}\bra{\phi}\right),
		\label{eq:lemma_product}
\end{equation}	
for any $n$-partite pure product state $\ket{\phi}$.
\end{lemma}

We now use Theorem~\ref{thm:TER} and the above lemmas 
to obtain the stability theorem of the output quantum $p$-R\'{e}nyi entropy 
for the depolarizing channel.
%
%
\begin{lemma}
\label{lemma:STRenyi}
Let $p \ge 2$, $\epsilon>0$ and $|\psi\rangle\in(\mathbb{C}^d)^{\otimes n}$ be a state. Then
\begin{equation}
	S_p\left(\mathcal{D}_{\lambda}^{\otimes n}\ket{\psi}\bra{\psi}\right)
		< S_p^{\min}\left( \mathcal{D}_{\lambda}^{\otimes n}\right)
			+2\epsilon\frac{p}{p-1}
				\frac{r-1}{r+1}
					\frac{(r^{p-1}-1)(2r^{p}+dr+d-2)}{(r^{p}+d-1)^2}
					+ O(\epsilon^{3/2})
\label{eq:True}
\end{equation}
holds only if a pure product state~$\ket{\phi}$ exists such that~$\ket{\psi}$ satisfies
\begin{equation}
\label{eq:close}
	\left|\left\langle\psi\right|\phi\rangle\right|^{2} \ge 1-\epsilon.
\end{equation}
\end{lemma}
\begin{proof}
We prove the contrapositive of the theorem. 
Let
\begin{equation}
\epsilon_0 
= 1-\max\left\{ \left| \left\langle\psi\right| \phi_{1},\cdots,\phi _{n}\rangle\right|^{2}:
			\left| \phi_{i}\right\rangle \in \mathbb{C}^d\right\}>\epsilon.
\label{eq:MaxFidelity}
\end{equation} 
Without loss of generality, 
we may assume that 
one of the states achieving the maximum in Eq.~(\ref{eq:MaxFidelity}) is $\ket{0^n}=\ket{0}$. 
We then have
\begin{equation}
	\ket{\psi}=\sqrt{1-\epsilon_0}\ket{0}+\sqrt{\epsilon_0}\ket{\phi}
	\label{eq:main_Thm1}
\end{equation}
for some state $\ket{\phi}$ such that $\inn{0}{\phi}=0$;
that is, $\ket{\phi}=\sum_{x\neq0}\alpha_x\ket{x}$ 
for some $\alpha_x $ such that $\sum_{x\neq0}\left|\alpha_x\right|^2=1$.
We can write explicitly
\begin{equation}
\ket{\psi}\bra{\psi}
=\left(1-\epsilon_0\right)\ket{0}\bra{0}+\sqrt{\epsilon_0\left(1-\epsilon_0\right)}
\left( \ket{0}\bra{\phi}+\ket{\phi}\bra{0}\right)+\epsilon_0\ket{\phi}\bra{\phi}.
\label{eq:main_Thm2}
\end{equation}
Therefore, we have
\begin{align}
	\mathcal{D}_{\lambda}^{\otimes n}\ket{\psi}\bra{\psi}
		=& \mathcal{D}_{\lambda}^{\otimes n}
			\ket{0}\bra{0}+\sqrt{\epsilon_0}\sqrt{1-\epsilon_0}~\mathcal{D}_{\lambda}^{\otimes n}
				\left(\ket{0}\bra{\phi}+\ket{\phi}\bra{0}\right)
					+\epsilon_0~\mathcal{D}_{\lambda}^{\otimes n}\left(\ket{\phi}\bra{\phi}
						-\ket{0}\bra{0}\right) \nonumber\\
		=& \mathcal{D}_{\lambda}^{\otimes n}
			\ket{0}\bra{0}+\sqrt{\epsilon_0}~\mathcal{D}_{\lambda}^{\otimes n}
				\left(\ket{0}\bra{\phi}+\ket{\phi}\bra{0}\right)
					+\epsilon_0~\mathcal{D}_{\lambda}^{\otimes n}\left(\ket{\phi}\bra{\phi}
						-\ket{0}\bra{0}\right)+O( \epsilon_0^{3/2} ).
\label{eq:main_Thm3}
\end{align}
For the last equality in Eq.~(\ref{eq:main_Thm3}), 
we use the Taylor expansion
\begin{equation}
	\sqrt{1-x}=1-\frac{1}{2}x+O\left( x^2\right)
\end{equation}
for all $|x|<1$.
Now we use 
\begin{align}
\rho &= \mathcal{D}_{\lambda}^{\otimes n}\ket{0}\bra{0}, \nonumber \\
\gamma_0 &= \mathcal{D}_{\lambda}^{\otimes n}\left(\ket{0}\bra{\phi}+\ket{\phi}\bra{0}\right), \nonumber \\
\gamma_1&= \mathcal{D}_{\lambda}^{\otimes n}\left(\ket{\phi}\bra{\phi} -\ket{0}\bra{0}\right), \nonumber \\
t&=\sqrt{\epsilon_0},  \nonumber \\
\rho\left(t\right) &= \rho+t \gamma_0+t^2 \gamma_1+O\left(t^3\right). \nonumber \\
\end{align}
Then Theorem~\ref{thm:TER} implies that
\begin{align}
	S_p\left(\mathcal{D}_{\lambda}^{\otimes n}\ket{\psi}\bra{\psi}\right) 
		=&S_p\left( \rho\left(t\right)\right)
						\nonumber\\
		=&S_p\left(\rho\right)+\frac{1}{1-p}t^2
			\left(p\frac{\T\left(\rho^{p-1}\gamma_1\right)}
				{\T\left(\rho^{p}\right)}+\frac{\T\left(Q_{\rho}\left(\gamma_0\right)\right)}
					{\T\left(\rho^{p}\right)}\right)
						+O\left(t^3\right).
\label{eq:RentropyPerturbation}
\end{align}
For convenience, we let
\begin{equation}
	a=\left(1+(d-1)\lambda\right)/d,\;
	b=(1-\lambda)/d.
\end{equation}
Then we obtain the following three facts.
\begin{enumerate}
	\item	As $\rho=\mathcal{D}_{\lambda}^{\otimes n}\ket{0}\bra{0}$ can be rewritten as
$\sum_{y} a^{n-|y}b^{|y|}\ket{y}\bra{y}$, 
	\begin{equation}
		\T\left( \rho^{p}\right)
			=\T\left(\sum_{y}\left(a^{p}\right)^{n-|y|}\left(b^{p}\right)^{|y|}\ket{y}\bra{y}\right)
			=\left( a^p+(d-1)b^p\right)^n,
	\end{equation}
	for~$|y|$ denoting Hamming weight of an $n$-bit string~$y$. 
	\item We can be evaluate
	\begin{align}
		\T\left(\gamma_1\rho^{p-1}\right) 
		=& \T\left(\mathcal{D}_{\lambda}^{\otimes n}
			\ket{\phi}\bra{\phi}\right)\left( \mathcal{D}_{\lambda}^{\otimes n}
				\ket{0}\bra{0}\right)^{p-1} -\T\left( \rho^p\right) \nonumber\\
		=& \sum_{x,x' \neq 0,y } \alpha_x {\alpha}_{x'}^*
			\left(a^{p-1}\right)^{n-|y|}\left( b^{p-1}\right)^{|y|}
				\T\left( \left(\mathcal{D}_{\lambda}^{\otimes n}
					\ket{x}\bra{x'}\right)\ket{y}\bra{y}\right) -\T\left( \rho^p\right)\nonumber\\
		=& \sum_{x \neq 0,y} |\alpha_x|^2\left(a^{p-1}\right)^{n-|y|}\left( b^{p-1}\right)^{|y|} \prod_{i=1}^n \frac{1-\lambda(1-d)^{\delta_{x_i y_i}}}{d}-\T\left( \rho^p\right) \nonumber\\
		=& \sum_{x \neq 0} |\alpha_x|^2\left( a^p+(d-1)b^p\right)^n
		\left(\frac{ a^{p-1}b+ab^{p-1}+(d-2)b^p}{a^p+(d-1)b^p}\right)^{|x|}
		-\T\left( \rho^p\right).
		\label{eq:main_Thm4}
\end{align}
	\item For $n$-bit strings $j$ and $k$, let 
	\begin{equation}
		g_{jk}
			:= \frac{\left(a^{n-|j|}b^{|j|}\right)^{p-1}-\left(a^{n-|k|}b^{|k|}\right)^{p-1}}{2\left(a^{n-|j|}b^{|j|}-a^{n-|k|}b^{|k|}\right)}.
	\label{eq:main_Thm5}
\end{equation}
Then we write
\begin{align}
	\T\left( Q_\rho\left( \gamma_0\right)\right) 
		=& p\left( \sum_{jk} g_{jk}\right) \left|\left( \mathcal{D}_{\lambda}^{\otimes n}
			\ket{0}\bra{\phi}\right)_{jk}\right|^2 
			\label{eq:main_Thm6} \\
		=& p \sum_{x \neq 0} |{\alpha}_x|^2 
			\left(\frac{(a^{|x|})^{p-1}-(b^{|x|})^{p-1}}{a^{|x|}-b^{|x|}}\right)
				(a^p+(d-1)b^p)^{n-|x|}.
				\label{eq:main_Thm7}
\end{align}
\end{enumerate}
Here $Q_\rho$~(\ref{eq:main_Thm6}) is a polynomial defined in Eq.~(\ref{eq:lemma_GF1}), and 
all equalities can be proved by tedious but straightforward calculations 
except the last equality in Eq.~(\ref{eq:main_Thm4}), 
which can be shown by mathematical induction on $n$.

Combining the above facts, we have
\begin{equation}
	S_p\left(\rho\left(t\right)\right)-S_p\left(\rho\right)
		=  t^2 \sum_{x \neq 0} |{\alpha_x}|^2 f_p(|x|)+O(t^3)
		\label{eq:main_Thm8}
\end{equation}
with 
\begin{align}
	f_p(|x|)
		=&\frac{p}{1-p}
			\left(\frac{(a^{|x|})^{p-1}-(b^{|x|})^{p-1}}{a^{|x|}-b^{|x|}}\right)
				\left(\frac{\lambda^2}{a^p+(d-1)b^p}\right)^{|x|}
							\nonumber\\
					&+\frac{p}{1-p}\left[\left(\frac{a^{p-1}b
						+ab^{p-1}+(d-2)b^p}{a^p+(d-1)b^p}\right)^{|x|} -1\right].
						\label{eq:main_Thm9}
\end{align}
Here it can be shown that 
the function $f_p$~(\ref{eq:main_Thm9}) is equal to 
the function $f_p$ defined in Eq.~(\ref{eq:falphax})
by taking $r=a/b$.  
Hence, Lemma~\ref{lem:f_p(x)} implies that 
$f_p(|x|)$ is monotonically increasing on $|x|$ for all $p \geq 2$.
As $\ket{\phi}$ does not have any weight-one components~\cite{HM13};
that is, $\alpha_x=0$ for $|x|< 2$, from Theorem \ref{thm:TER} and Lemma \ref{lem:minRenyi},
we can finally obtain the following inequality, 
and thereby complete the proof.
\begin{align}
	S_p\left(\mathcal{D}_{\lambda}^{\otimes n}\ket{\psi}\bra{\psi}\right)-
		S_p^{\min}\left(\mathcal{D}_{\lambda}^{\otimes n}\right)
			=&S_p\left(\rho\left(t\right)\right)-S_p\left(\rho\right) 
							\nonumber\\
			\ge& \epsilon_0 f_p(2)+O(\epsilon_0^{3/2})
\label{eq:ineq}
							\nonumber\\
			>&\epsilon \frac{p}{p-1}
				\left[\frac{2\lambda(1-\lambda)(a^{p-1}-b^{p-1})
					(2(a^{p}-b^{p})+db^{p-1}(a+b))}
						{(2+(d-2)\lambda)(a^{p}+(d-1)b^p )^2}\right] +O(\epsilon^{3/2})
							\nonumber\\
		=& 2\epsilon\frac{p}{p-1}
			\frac{(r-1)(r^{p-1}-1)(2r^{p}+dr+d-2)}{(r+1)(r^{p}+d-1)^2}
				+O(\epsilon^{3/2}).
\end{align}
\end{proof}
\begin{remark}
Although we have not yet established the stability theorem for $1<p<2$,
we can show that it still holds for $p=1$;
that is, the stability theorem for the von~Neumann entropy of the depolarizing channel holds
by using a similar method to what we have used.
(We have numerically checked that the same result holds 
for the several cases of $p$ with $1<p<2$.)
Let us see the proof for the case of $p=1$.

We calculate the Taylor expansion of the von~Neumann entropy 
to get the difference between the output von~Neumann entropy and its minimum value as follows.
\begin{equation}
S\left(\mathcal{D}_{\lambda}^{\otimes n}\ket{\psi}\bra{\psi}\right)
-S^{\min}\left(\mathcal{D}_{\lambda}^{\otimes n}\right)
=\epsilon\sum_{x \neq 0} |{\alpha_x}|^2 f(|x|)+O(\epsilon^{3/2}),
\label{eq:output_von_Neumann_entropy}
\end{equation}
where the function $f$ is defined as
\begin{equation}
f(|x|)= |x|(a-b) \log \frac{a}{b}-\left(a-b\right)^{2|x|}\left(\frac{\log a^{|x|}-\log b^{|x|}}{a^{|x|}-b^{|x|}}\right),
\end{equation}
which is the limit of $f_p(|x|)$ for $p$ tends to one. 
As it is easier than Lemma~\ref{lem:f_p(x)} 
to prove that $f(|x|)$ is monotonically increasing, 
we can easily obtain the almost same result as 
our stability theorem.
\end{remark}

\begin{theorem}
\label{thm:stability_new}
The depolarizing channel is stable with respect to the output quantum $p$-R\'{e}nyi entropy 
for $p \ge 2$ or $p=1$.
\end{theorem}
\begin{proof}
Let $\epsilon>0$ be given,
and let the quantity $Q$ be~$S_p$.
Let $\ket{\psi}\bra{\psi}$ be
an~$\epsilon$-almost extremal state with respect to~$S_p$ and $\mathcal{D}_{\lambda}^{\otimes n}$.
Then
\begin{equation}
	\left|S_p\left(\mathcal{D}_{\lambda}^{\otimes n}\ket{\psi}\bra{\psi}\right)
		-S_p^{\min}\left(\mathcal{D}_{\lambda}^{\otimes n}\right)\right|
			\in O(\epsilon),
\end{equation}
by Lemma~\ref{lemma:STRenyi}, 
there exists some extremal state $\ket{\phi}\bra{\phi}$ 
with respect to~$S_p$ and $\mathcal{D}_{\lambda}^{\otimes n}$ such that
\begin{equation}
\left\| \ket{\psi}\bra{\psi}-\ket{\phi}\bra{\phi} \right\|^2_1
\le \epsilon.
\end{equation}
Thus, the depolarizing channel $\mathcal{D}_{\lambda}^{\otimes n}$ is~$\epsilon$-stable with respect to the quantum $p$-R\'{e}nyi entropy~$S_p$
for any $\epsilon>0$, 
and hence it is stable with respect to~$S_p$.
\end{proof}

\begin{remark}
	We obtain the stability theorem of the depolarizing channel with respect to the output purity~\cite{HM13}
	as a corollary of Theorem~\ref{thm:stability_new}.
\end{remark}

Furthermore, we can similarly show that if $n$-qudit pure state is close to be product 
then its output quantum $p$-R\'{e}nyi entropy is close to 
the minimal output quantum $p$-R\'{e}nyi entropy with a specific precision as follows.  
\begin{theorem}
\label{theorem:STRenyi2}
Let $p \ge 2$, $\epsilon>0$ and $|\psi\rangle\in(\mathbb{C}^d)^{\otimes n}$ be a state. Then
\begin{equation}
\label{eq:False}
	S_p\left(\mathcal{D}_{\lambda}^{\otimes n}\ket{\psi}\bra{\psi}\right)
		\geq S_p^{\min}\left( \mathcal{D}_{\lambda}^{\otimes n}\right)
			+\epsilon\frac{p}{p-1}+ O(\epsilon^{3/2})
\end{equation}
implies 
\begin{equation}
\left| \inn{\psi}{\phi}\right|^{2}<1-\epsilon
\label{eq:False2}
\end{equation}
 for any product state $|\phi\rangle$.
\end{theorem}
\begin{proof}
Suppose that
\begin{equation}
\label{eq:tau}
1-\epsilon_1=\max\left\{\left| \left\langle\psi\right|\phi_{1},\cdots,\phi _{n}\rangle\right|^{2}:
			\ket{\phi_{i}} \in \mathbb{C}^d\right\} \ge 1-\epsilon.
\end{equation}
From the same arguments in Theorem~\ref{lemma:STRenyi}, 
we obtain the same equality as in Eq.~(\ref{eq:main_Thm8}).
Then
\begin{align}
	S_p\left(\mathcal{D}_{\lambda}^{\otimes n}\ket{\psi}\bra{\psi}\right)-
		S_p^{\min}\left(\mathcal{D}_{\lambda}^{\otimes n}\right)
			=&  \epsilon_1 \sum_{x \neq 0} |{\alpha_x}|^2 f_p(|x|)+O(\epsilon_1^{2/3})
							\nonumber\\
			<& \epsilon_1\frac{p}{1-p}+O(\epsilon_1^{3/2})
							\nonumber\\
			\le& \epsilon \frac{p}{1-p}+O(\epsilon^{3/2}),
\label{eq:epsilon}	
\end{align}
where the first inequality is obtained due to the monotonicity of $f_p(|x|)$
and the second inequality results from Eq.~(\ref{eq:tau}). 
\end{proof}

\begin{remark}
The coefficient of~$\epsilon$ in Eq.~(\ref{eq:True}) is smaller than 
the coefficient of~$\epsilon$ in Eq.~(\ref{eq:False}),
which means that some gap exists between them 
even though it is close to zero for sufficiently small~$\epsilon$. 
Furthermore, for a sufficiently large $p$, 
the gap can be smaller than the gap for the case of $p=2$ 
as we will see in \S\ref{sec:Poly}.
\end{remark}

\section{An application: A Polygraph Test}
\label{sec:Poly}
In this section, we introduce a polygraph test 
as an application of our main results,
namely Theorems~\ref{lemma:STRenyi} and~\ref{theorem:STRenyi2}.
Let us consider the following protocol 
wherein sender Alice transmits multiple copies of an $n$-qudit state 
through depolarizing channels to receiver Bob.
\renewcommand{\labelenumi}{\arabic{enumi}.}
\begin{enumerate} 
	\item Bob informs Alice of a small enough $\epsilon>0$ chosen as an error bound.
	\item Alice prepares an $n$-qudit pure state 
that is close to a product state with fidelity at least $1-\epsilon$ as in Eq.~(\ref{eq:close})
and sends multiple copies to Bob through depolarizing channels.
	\item Bob estimates its output quantum R\'{e}nyi entropy.
	\item Bob determines whether Alice's preparation satisfies the requirement or not,
		and our results help Bob make the correct decision as discussed below.
		\begin{itemize}
			\item[Accept:]
				If Bob's estimate of the output quantum R\'{e}nyi entropy 
				satisfies Inequality~(\ref{eq:True}) then 
				Alice definitely prepared a correct state according to
				Theorem~\ref{lemma:STRenyi}. 
			\item[Reject:]
				If Bob's estimate of the output quantum R\'{e}nyi entropy 
				satisfies Inequality~(\ref{eq:False}) then
				Theorem~\ref{theorem:STRenyi2} guarantees
				that Alice's preparation fails the requirement.
		\end{itemize}
\end{enumerate}
\begin{remark}
	Some gap exists between the coefficients of~$\epsilon$ in Eq.~(\ref{eq:True}) 
	and of~$\epsilon$ in Eq.~(\ref{eq:False}), 
	which means that Bob cannot detect Alice's lie 
	when neither Eqs.~(\ref{eq:True}) nor~(\ref{eq:False}) holds
	for the output quantum R\'{e}nyi entropy of the state Alice sent. 
	However, the probability that Alice cheats Bob can be forced to be close to zero 
	if Bob chooses small enough~$\epsilon$. 
	Thus, the gap problem can be resolved in this way.
\end{remark}

We note that the above polygraph test can be also realized 
by the original stability theorem~\cite{HM13} in the same way as ours,
as the original stability theorem is essentially equivalent to 
the case of $p=2$ in ours.
However, we show that 
if $p$ is sufficiently large then 
our gap can be smaller than the gap from the original stability as follows.
Let a nonzero~$\epsilon$ be fixed, 
and define the gap function
\begin{equation}
	\operatorname{gap}(p)
		:=\frac{p}{p-1}\left(1-\frac{2(r-1)(r^{p-1}-1)(2r^p+dr+(d-2))}{(r+1)(r^p+(d-1))^2}\right),
\end{equation}
which is the gap between the coefficients of~$\epsilon$ 
in Eqs.~(\ref{eq:True}) and~(\ref{eq:False}). 
We claim 
\begin{equation}
	\operatorname{gap}(2)>\lim_{p \rightarrow \infty}\operatorname{gap}(p), 
\label{eq:gap}
\end{equation}
which is equivalent to the inequality
\begin{equation}
	 (r^2+(d-1))^2(r^2+5r-4) > 4r(r-1)^2(2r^2+dr+(d-2)).
 \label{eq:gap2}
\end{equation}
Moreover, in order to prove the inequality (\ref{eq:gap2})
it is enough to show that
\begin{equation}
	h(r)
		:=r^3-3r^2+(-2d+10)r+(14d-10)
\label{eq:gap3}
\end{equation}
is positive for all $r>1$.
As its derivative is 
\begin{equation}
h'(r)=3r^2-6r+(-2d+10),
\label{eq:h2}
\end{equation}
$h(r)$ is evidently positive if $d$ is 2 or 3. 
If $d\ge4$ then it can be readily shown that
\begin{equation}
h(r)\ge h\left(1+\frac{\sqrt{6d-21}}{3}\right)=\frac{2\sqrt{6d-21}}{3}-2+12d>0
\label{eq:gap4}
\end{equation} 
for all $r>1$.

We have introduced the polygraph test as an application of out main theorem,
and we have proved that the protocol for our stability theorem has a smaller undecidable gap than for the protocol in the original stability theorem.

\section{Conclusions}
\label{sec:Conclusions}
We have shown that the stability theorem of the depolarizing channel holds 
for the output quantum $p$-R\'{e}nyi entropy for $p \ge 2$, 
which was one of the open questions in Ref.~\cite{HM13}. 
 Furthermore, we have also proved that the stability theorem of the depolarizing channel holds for the output von~Neumann entropy ($p=1$), 
and have numerically checked that 
the stability theorem 
holds 
for the several cases $p$ with $1<p<2$.
Therefore, we expect that the stability theorem 
holds 
for all $p \ge 1$, and 
leave this for future work.

As an application of our main results, we have introduced a polygraph test
and have presented its protocol. 
The original stability theorem can be also applied to the polygraph test, 
as the original one is essentially equivalent to our stability theorem
when $p=2$. 
In the protocol, Bob determines 
whether Alice prepares a pure quantum state close to a product state.
However, Bob cannot perfectly decide whether her preparation is proper, 
that is, there is an undecidable gap in which 
he can decide nothing. 
We have shown that the undecidable gap of our protocol 
can be smaller than the original case.
Therefore, our results improve the original stability theorem as well as generalize it.


\acknowledgements{%
We acknowledge discussions with A. W. Harrow and A. Montanaro.
This research was supported by the Basic Science Research Program through the National Research Foundation of Korea funded by the Ministry of Education (NRF-2012R1A1A2003441).
E. B.\ and B. C. S.\ acknowledge financial support from AITF and NSERC,
and B.C.S.\ acknowledges support from China's 1000 Talent Plan.
G.G.\ acknowledges financial support from NSERC.%
}

\bibliography{stability}

\end{document}